\documentclass[letterpaper, 12 pt, conference]{ieeeconf}  

\IEEEoverridecommandlockouts                              

	\usepackage{cite}

	\usepackage[cmex10]{amsmath}
	\usepackage{mathtools}

	\usepackage{array}
	\usepackage{mdwmath}
	\usepackage{mdwtab}

	\usepackage{url}
	
	\usepackage{graphics} 
	\usepackage{epsfig} 
	\usepackage{amssymb}  
	\usepackage{hyperref} 
	\usepackage{subfig}
	\usepackage{verbatim}
	\usepackage{color}
	\usepackage{multirow} 
\usepackage[utf8]{inputenc}
\usepackage[T1]{fontenc}

\usepackage{setspace}
	\newtheorem{theorem}{Theorem}
	
	\newtheorem{lemma}{Lemma}

	\newcommand{\tr}{\dagger}
\newcommand{\eq}[1]{\begin{align}#1\end{align}}
\newcommand{\seq}[1]{\begin{subequations}#1\end{subequations}}

\newcommand{\bm}[1]{\begin{bmatrix}#1\end{bmatrix}}
\newcommand{\bit}[1]{\begin{itemize}#1\end{itemize}}
\newcommand{\E}{\mathbb{E}}
\newcommand{\R}{\mathbb{R}}
\newcommand{\cT}{\mathcal{T}}
\newcommand{\cP}{\mathcal{P}}
\newcommand{\cX}{\mathcal{X}}
\newcommand{\cU}{\mathcal{U}}
\newcommand{\cH}{\mathcal{H}}
\newcommand{\bV}{\mathbf{V}}
\newcommand{\bF}{\mathbf{F}}
\newcommand{\bT}{\mathbf{T}}
\newcommand{\bP}{\mathbf{P}}
\newcommand{\bS}{\mathbf{S}}
\newcommand{\bL}{\mathbf{L}}
\newcommand{\bI}{\mathbf{I}}
\newcommand{\bJ}{\mathbf{J}}
\newcommand{\bG}{\mathbf{G}}

\newcommand{\bSigma}{\mathbf{\Sigma}}
\newcommand{\bVb}{\bar{\mathbf{V}}}
\newcommand{\bD}{\mathbf{D}}
\newcommand{\bC}{\mathbf{C}}
\newcommand{\bA}{\mathbf{A}}
\newcommand{\bR}{\mathbf{R}}
\newcommand{\bB}{\mathbf{B}}
\newcommand{\bQ}{\mathbf{Q}}

\newcommand{\bLambda}{\mathbf{\Lambda}}
\newcommand{\bDelta}{\mathbf{\Delta}}
\newcommand{\bZero}{\mathbf{0}}
\newcommand{\bOne}{\mathbf{1}}
\newcommand{\btL}{\mathbf{\tilde{L}}}
\newcommand{\btD}{\mathbf{\tilde{D}}}
\newcommand{\btG}{\mathbf{\tilde{G}}}
\newcommand{\btC}{\mathbf{\tilde{C}}}
\newcommand{\btF}{\mathbf{\tilde{F}}}
\newcommand{\tm}{\tilde{m}}
\newcommand{\btM}{\mathbf{\tilde{M}}}
\newcommand{\bbB}{\mathbf{\bar{B}}}
\newcommand{\bbL}{\mathbf{\bar{L}}}
\newcommand{\bbA}{\mathbf{\bar{A}}}

\newcommand{\pushright}[1]{\ifmeasuring@ #1 \else\omit\hfill$\displaystyle#1$\fi\ignorespaces}
\newcommand{\pushleft}[1]{\ifmeasuring@ #1 \else\omit$\displaystyle#1$\hfill\fi\ignorespaces}
\newcounter{exercisenumber}
\begin{document}
	\onecolumn 
	%
	\title{Signaling equilibria for dynamic LQG games with asymmetric information  }
	%
	%
	%
	\author{Deepanshu~Vasal and Achilleas~Anastasopoulos
	\thanks{The authors are with the Department
	of Electrical Engineering and Computer Science, University of Michigan, Ann
	Arbor, MI, 48105 USA e-mail: { \{dvasal, anastas\} at umich.edu}}
	}

	\maketitle
\begin{abstract}
We consider a finite horizon dynamic game with two players who observe their types privately and take actions, which are publicly observed. Players' types evolve as independent, controlled linear Gaussian processes and players incur quadratic instantaneous costs. This forms a dynamic linear quadratic Gaussian (LQG) game with asymmetric information. 
We show that under certain conditions, players' strategies that are linear in their private types, together with Gaussian beliefs form a perfect Bayesian equilibrium (PBE) of the game. 
Furthermore, it is shown that this is a signaling equilibrium due to the fact that future beliefs on players' types are affected by the equilibrium strategies.
We provide a backward-forward algorithm to find the PBE. Each step of the backward algorithm reduces to solving an algebraic matrix equation for every possible realization of the state estimate covariance matrix.
The forward algorithm consists of Kalman filter recursions, where state estimate covariance matrices depend on equilibrium strategies.
\end{abstract}

\section{Introduction}

Linear quadratic Gaussian (LQG) team problems have been studied extensively under the framework of classical stochastic control with single controller and perfect recall~\cite[Ch.7]{KuVa86}. In such a system, the state evolves linearly and the controller makes a noisy observation of the state which is also linear in the state and noise. The controller incurs a quadratic instantaneous cost.  With all basic random variables being independent and Gaussian, the problem is modeled as a partially observed Markov decision process (POMDP). The belief state process under any control law happens to be Gaussian and thus can be sufficiently described by the corresponding mean and covariance processes, which can be updated by the Kalman filter equations. Moreover, the covariance can be computed offline and thus the mean (state estimate) is a sufficient statistic for control. Finally, due to the quadratic nature of the costs, the optimal control strategy is linear in the state. Thus, unlike most POMDP problems, the LQG stochastic control problem can be solved analytically and admits an easy-to-implement optimal strategy.   

LQG team problems have also been studied under non-classical information structure such as in multi-agent decentralized team problems where two controllers with different information sets minimize the same objective. Such systems with asymmetric information structure are of special interest today because of the emergence of large scale networks such as social or power networks, where there are multiple decision makers with local or partial information about the system. It is well known that for decentralized LQG team problems, linear control policies are not optimal in general~\cite{Wi68}. However there exist special information structures, such as partially nested~\cite{HoCh72} and stochastically nested~\cite{Yu09}, where linear control is shown to be optimal. Furthermore, due to their strong appeal for ease of implementation, linear strategies have been studied on their own for decentralized teams even at the possibility of being suboptimal (see~\cite{MaNa15} and references therein).
 
When controllers (or players) are strategic, the problem is classified as a dynamic game and an appropriate solution concept is some notion of equilibrium. 
When players have different information sets, such games are called games with asymmetric information. There are several notions of equilibrium for such games, including perfect Bayesian equilibrium (PBE), sequential equilibrium, trembling hand equilibrium \cite{OsRu94,FuTi91}. Each of these notions of equilibrium consists of a strategy and a belief profile of all players where the equilibrium strategies are optimal given the beliefs and the beliefs are derived from the equilibrium strategy profile and using Bayes' rule (whenever possible), with some equilibrium concepts requiring further refinements. Due to this circular argument of beliefs being consistent with strategies which are in turn optimal given the beliefs, finding such equilibria is a difficult task. To date, there is no known sequential decomposition methodology to find such equilibria for general dynamic games with asymmetric information.

Authors in~\cite{Ba78} studied a discrete-time dynamic LQG game with one step delayed sharing of observations. 
Authors in~\cite{NaGuLaBa14} studied a class of dynamic games with asymmetric information under the assumption that player's posterior beliefs about the system state conditioned on their common information are independent of the strategies used by the players in the past. 
Due to this independence of beliefs and past strategies, the authors of~\cite{NaGuLaBa14} were able to provide a backward recursive algorithm similar to dynamic programming to find Markov perfect equilibria~\cite{MaTi01} of a transformed game which are equivalently a class of Nash equilibria of the original game. The same authors specialized their results in~\cite{GuNaLaBa14} to find non-signaling equilibria of dynamic LQG games with asymmetric information. 

Recently, we considered a general class of dynamic games with asymmetric information and independent private types in \cite{VaAn16} and provided a sequential decomposition methodology to find a class of PBE of the game considered. In our model, beliefs depend on the players' strategies, so 
our methodology allows the possibility of finding signaling equilibria.
In this paper, we build on this methodology to find signaling equilibria for two-player dynamic LQG games with asymmetric information. We show that players' strategies that are linear in their private types in conjunction with consistent Gaussian beliefs form a PBE of the game. Our contributions are:
\bit{
\item[(a)] Under strategies that are linear in players' private types, we show that the belief updates are Gaussian and the corresponding mean and covariance are updated through Kalman filtering equations which depend on the players' strategies, unlike the case in classical stochastic control and the model considered in~\cite{GuNaLaBa14}. Thus there is signaling~\cite{Ho80,KrSo94}.

\item[(b)] We sequentially decompose the problem by specializing the forward-backward algorithm presented in~\cite{VaAn16} for the dynamic LQG model. The backward algorithm requires, at each step, solving a fixed point equation in `partial' strategies of the players for all possible beliefs. We show that in this setting, solving this fixed point equation reduces to solving a matrix algebraic equation for each realization of the state estimate covariance matrices. 

\item[(c)] The cost-to-go value functions are shown to be quadratic in the private type and state estimates, which together with quadratic instantaneous costs and mean updates being linear in the control action, implies that at every time $t$ player $i$ faces an optimization problem which is quadratic in her control. Thus linear control strategies are shown to satisfy the optimality conditions in~\cite{VaAn16}.

\item[(d)] For the special case of scalar actions, we provide sufficient algorithmic conditions for existence of a solution of the algebraic matrix equation. Finally, we present numerical results on the steady state solution for specific parameters of the problem.

}
  
The paper is structured as follows. In Section~\ref{sec:model}, we define the model. In Section~\ref{sec:prelim}, we introduce the solution concept and summarize the general methodology in~\cite{VaAn16}. In Section~\ref{sec:results}, we present our main results where we construct equilibrium strategies and belief through a forward-backward recursion. In Section~\ref{sec:disc} we discuss existence issues and present numerical steady state solutions. We conclude in Section~\ref{sec:concl}.

\subsection{Notation}
We use uppercase letters for random variables and lowercase for their realizations. We use bold upper case letters for matrices. For any variable, subscripts represent time indices and superscripts represent player identities. We use notation $ -i$ to represent the player other than player $i$. We use notation $a_{t:t'}$ to represent vector $(a_t, a_{t+1}, \ldots a_{t'})$ when $t'\geq t$ or an empty vector if $t'< t$. We remove superscripts or subscripts if we want to represent the whole vector, for example $ a_t$  represents $(a_t^1, a_t^2) $. We use $\delta(\cdot)$ for the Dirac delta function. We use the notation $X\sim F$ to denote that the random variable $X$ has distribution $F$.
For any Euclidean set $\mathcal{S}$, $\cP(\mathcal{S})$ represents the space of probability measures on $\mathcal{S}$ with respect to the Borel sigma algebra. We denote by $P^g$ (or $\E^g$) the probability measure generated by (or expectation with respect to) strategy profile $g$. We denote the set of real numbers by $\R$. 
For any random vector $X$ and event A, we use the notation $sm(\cdot|\cdot)$ to denote the conditional second moment, $sm(X|A) : = \E[XX^{\tr}|A]$. For any matrices $\bA$ and $\bB$, we will also use the notation $quad(\cdot;\cdot)$ to denote the quadratic function, $quad(\bA;\bB) :=\bB^{\tr}\bA\bB$. We denote trace of a matrix $\bA$ by $tr(\bA)$. $N(\hat{x},\bSigma)$ represents the vector Gaussian distribution with mean vector $\hat{x}$ and covariance matrix $\bSigma$.
All equalities and inequalities involving random variables are to be interpreted in \emph{a.s.} sense and inequalities in matrices are to be interpreted in the sense of positive definitedness. All matrix inverses are interpreted as pseudo-inverses.

\section{Model}
\label{sec:model}

We consider a discrete-time dynamical system with 2 strategic players over a finite time horizon $\mathcal{T} := \{1, 2, \ldots T\}$ and with perfect recall. There is a dynamic state of the system $x_t := (x_t^1, x_t^2)$, where $x_t^i \in \cX^i :=\R^{n_i}$ is private type of player $i$ at time $t$ which is perfectly observed by her. Player $i$ at time $t$ takes action $u_t^i \in \cU^i := \R^{m_i}$ after observing $u_{1:t-1}$, which is common information between the players, and $x_{1:t}^i$, which it observes privately. Thus at any time $t\in \cT$, player $i$'s information is $u_{1:t-1},x_{1:t}^i$. Players' types evolve linearly as 
\eq{
x_{t+1}^i = \bA_t^ix_t^i + \bB_t^i u_t + w_t^i, \label{eq:sys_evol} 
}
where $\bA_t^i, \bB_t^i$ are known matrices. $(X^1_1,X_1^2,(W_t^i)_{t\in\cT})$ are basic random variables of the system which are assumed to be independent and Gaussian such that $X_1^i\sim N(0,\bSigma^i_1)$ and $W_t^i \sim N(0,\bQ^i)$. As a consequence, types evolve as conditionally independent, controlled Markov processes,
\eq{
P(x_{t+1}|u_{1:t},x_{1:t}) &= P(x_{t+1}|u_{t},x_{t} )=\prod_{i=1}^2 Q^i(x_{t+1}^i|u_{t},x_{t}^i ) .
}
where $Q^i(x_{t+1}^i|u_{t},x_{t}^i ):= P(w_t^i =x_{t+1}^i -\bA_t^ix_t^i - \bB_t^i u_t )$.
At the end of interval $t$, player $i$ incurs an instantaneous cost $R^i(x_t,u_t)$, 
\eq{
R^i(x_t,u_t) &= u_t^{\tr}\bT^iu_t + x_t^{\tr}\bP^ix_t +  2u_t^{\tr}\bS^ix_t \nonumber \\
&= \bm{u_t^{\tr} & x_t^{\tr}} \bm{\bT^i &\bS^i \\ \bS^{i\tr} & \bP^i}\bm{u_t \\ x_t},
}
where $\bT^i, \bP^i, \bS^i$ are real matrices of appropriate dimensions and $\bT^i,\bP^i$ are symmetric. We define the instantaneous cost matrix $\bR^i$ as $\bR^i := \bm{\bT^i &\bS^i \\ \bS^{i\tr} & \bP^i}$. Let $g^i = ( g^i_t)_{t \in \mathcal{T}}$ be a probabilistic strategy of player $i$,  where $g^i_t : (\cU^i)^{t-1}\times (\cX^i)^t \to \cP(\cU^i)$ such that player $i$ plays action $u_t^i$ according to distribution $ g^i_t(\cdot|u_{1:t-1},x_{1:t}^i)$. Let $ g :=(g^i)_{i=1,2}$ be a strategy profile of both players.
The distribution of the basic random variables and their independence structure together with the system evolution in (\ref{eq:sys_evol}) and players strategy profile $g$ define a joint distribution on all random variables involved in the dynamical process. The objective of player $i$ is to maximize her total expected cost 
\eq{ 
J^{i,g} := \E^g \left\{ \sum_{t=1}^T R^i(X_t,U_t) \right\} . \label{eq:J_obj}
}
With both players being strategic, this problem is modeled as a dynamic LQG game with asymmetric information and with simultaneous moves.

\section{Preliminaries}
\label{sec:prelim}
In this section we introduce the equilibrium concept for dynamic games with asymmetric information and summarize the general methodology developed in~\cite{VaAn16} to find a class of such equilibria.

\subsection{Solution concept}

Any history of this game at which players take action is of the form $h_t = (u_{1:t-1},x_{1:t})$. Let $\cH_t$ be the set of such histories at time $t$ and $\cH := \cup_{t=0}^T \cH_t $ be the set of all possible such histories. At any time $t$ player $i$ observes $h^i_t = (u_{1:t-1},x_{1:t}^{i})$ and both players together have $h^c_t = u_{1:t-1}$ as common history. Let $\cH^i_t$ be the set of observed histories of player $i$ at time $t$ and $\cH^c_t$ be the set of common histories at time $t$. An appropriate concept of equilibrium for such games is the PBE \cite{FuTi91} which consists of a pair $(\beta^*,\mu^*)$ of strategy profile $\beta^* = (\beta_t^{*,i})_{t \in \mathcal{T},i=1,2}$ where $\beta_t^{*,i} : \cH_t^i \to \cP(\cU^i)$ and a belief profile $\mu^* = (\mu_t^{*,i})_{t \in \mathcal{T},i=1,2}$ where $\mu_t^{*,i}: \cH^i_t \to \cP(\cH_t)$ that satisfy sequential rationality so that for $ i =1,2,  \forall t \in \mathcal{T},  h^{i}_t \in \cH^i_t, {\beta^{i}}$
\eq{
\E^{(\beta^{*,i} \beta^{*,-i},\, \mu^*)}\left\{ \sum_{n=t}^T R^i(X_n, U_n)\big\lvert  h^i_t\right\} \leq \E^{({\beta}^{i} \beta^{*,-i},\, \mu^*)}\left\{ \sum_{n=t}^T R^i(X_n, U_n)\big\lvert  h^i_t\right\}, \;\; \;\;   \label{eq:seqeq}
} 
and the beliefs satisfy consistency conditions as described in~\cite[p. 331]{FuTi91}.

\subsection{Structured perfect Bayesian equilibria}

A general class of dynamic games with asymmetric information was considered in \cite{VaAn16} by the authors where players' types evolve as conditionally independent controlled Markov processes. A backward-forward algorithm was provided to find a class of PBE of the game called structured perfect Bayesian equilibria (SPBE). In these equilibria, player $i$'s strategy is of the form $U_t^i \sim m_t^i(\cdot|\pi^1_t, \pi^2_t ,x_t^i)$ where $m_t^i: \cP(\cX^1) \times \cP(\cX^2) \times \cX^i \to \cP(\cU^i)$. Specifically, player $i$'s action at time $t$ depends on her private history $x_{1:t}^i$ only through $x_t^i$. Furthermore, it depends on the common information $u_{1:t-1}$ through a common belief vector $\underline{\pi}_t:=(\pi_t^1,\pi_t^2)$ where $\pi_t^i\in \cP(\cX^i)$ is belief on player $i$'s current type $x_t^i$ conditioned on common information $u_{1:t-1}$, i.e. $\pi_t^i(x_t^i):= P^g(X_t^i = x_t^i|u_{1:t-1})$. 

The common information $u_{1:t-1}$ was summarized into the belief vector $(\pi_t^1,\pi_t^2)$ following the common agent approach used for dynamic decentralized team problems \cite{NaMaTe13}. Using this approach, player $i$'s strategy can be equivalently described as follows: player $i$ at time $t$ observes $u_{1:t-1}$ and takes action $\gamma_t^i$, where $\gamma_t^i :  \cX^i \to \mathcal{P}(\cU^i)$ is a partial (stochastic) function from her private information $x_t^i$ to  $u_t^i$ of the form $U_t^i \sim \gamma_t^i(\cdot|x_t^i)$. These actions are generated through some policy $\psi^i = (\psi^i_t)_{t \in \mathcal{T}}$, $\psi^i_t : (\cU^i)^{t-1} \to \left\{  \cX^i \to \mathcal{P}(\cU^i) \right\}$, that operates on the common information $u_{1:t-1}$ so that $\gamma_t^i = \psi_t^i[u_{1:t-1}]$. Then any policy of the player $i$ of the form $U_t^i \sim g_t^i(\cdot|u_{1:t-1},x_t^{i})$ is equivalent to $U_t^i \sim \psi^i_t[u_{1:t-1}] (\cdot|x_t^i)$ \cite{NaMaTe13}. 

The common belief $\pi_t^i$ is shown in Lemma~2 of~\cite{VaAn16} to be updated as 
\seq{
\label{eq:pi_update}
\eq{
\pi_{t+1}^i(x_{t+1}^i) = 
 \frac{  \int_{x^i_t} \pi_t^{{i}}(x_t^i) \gamma_t^i(u_t^i|x_t^i)Q_t^i(x_{t+1}^i|x_t^i ,u_t) dx_t^i} %
	{ \int_{\tilde{x}_t^i}\pi_t^{{i}}(\tilde{x}_t^i)  \gamma_t^i(u_t^i|\tilde{x}_t^i)d\tilde{x}_t^i}, 
}
if the denominator is not 0, and as 
\eq{
\pi_{t+1}^i(x_{t+1}^i) =   \int_{x^i_t} \pi_t^{{i}}(x_t^i) Q_t^i(x_{t+1}^i|x_t^i ,u_t) dx_t^i,
} 
}
if the denominator is 0. The belief update can be summarized as,
\eq{
\pi^{i}_{t+1} = {\bar{F}}(\pi^{i}_t,\gamma^i_t,u_t),
}
where ${\bar{F}}$ is independent of players' strategy profile $g$. The SPBE of the game can be found through a two-step backward-forward algorithm. In the backward recursive part, an equilibrium generating function $\theta$ is defined based on which a strategy and belief profile ($\beta^*,\mu^*$) are defined through a forward recursion.

In the following we summarize the algorithm and results of~\cite{VaAn16}.

\subsubsection{Backward Recursion}

An equilibrium generating function $\theta=(\theta^i_t)_{i=1,2,t\in\mathcal{T}}$ and a sequence of value functions $(V_t^i)_{i=1,2, t\in \{ 1,2, \ldots T+1\}}$ are defined through backward recursion, where $\theta^i_t : \cP(\cX^1)\times \cP(\cX^2) \to \left\{\cX^i \to \cP(\cU^i) \right\}$, $V_t^i : \cP(\cX^1)\times \cP(\cX^2)  \times \cX^i \to \R$, as follows. 
\bit{

\item[(a)] Initialize $\forall \underline{\pi}_{T+1}\in \cP(\cX^1)\times \cP(\cX^2) , x_{T+1}^i\in \cX^i$,
\eq{
V^i_{T+1}(\underline{\pi}_{T+1},x_{T+1}^i) := 0.   \label{eq:VT+1}
}
}

\bit{
\item[(b)] For $t = T,T-1, \ldots 1$,

$ \forall \underline{\pi}_t \in \cP(\cX^1)\times \cP(\cX^2)  $, let $\theta_t[\underline{\pi}_t] $ be generated as follows. Set $\tilde{\gamma}_t = \theta_t[\underline{\pi}_t]$ where $\tilde{\gamma}_t$ is the solution of the following fixed point equation, $\forall i \in \mathcal{N},x_t^i\in \cX^i$,
  \eq{
 \tilde{\gamma}^{i}_t(\cdot|x_t^i) \in &\arg\min_{\gamma^i_t(\cdot|x_t^i)} \E^{\gamma^i_t(\cdot|x_t^i) \tilde{\gamma}^{-i}_t} \left\{ R^i(X_t,U_t) + V_{t+1}^i (F(\underline{\pi}_t, \tilde{\gamma}_t, U_t), X_{t+1}^i) \big\lvert  \underline{\pi}_t ,x_t^i \right\} , \label{eq:M_FP}
 }
 
 where expectation in (\ref{eq:M_FP}) is with respect to random variables $(x_t^{-i},u_t ,x_{t+1}^i)$ through the measure \\
${\pi}_t^{-i}(x_t^{-i})\gamma^i_t(u^i_t|x_t^i) \tilde{\gamma}^{-i}_t(u^{-i}_t|x_t^{-i})Q^i(x_{t+1}^i|x_t^i,u_t)$ and $F(\underline{\pi}_t,\gamma_t,u_t):=(\bar{F}(\pi_t^1,\gamma_t^1,u_t),\bar{F}(\pi_t^2,\gamma_t^2,u_t))$ .

Also define 
  \eq{
  V^i_{t}(\underline{\pi}_t,x_t^i) &:= \;\; \E^{\tilde{\gamma}^{i}_t(\cdot|x_t^i) \tilde{\gamma}^{-i}_t}\left\{ {R}^i (X_t,U_t) +V_{t+1}^i (F(\underline{\pi}_t, \tilde{\gamma}_t, U_t), X_{t+1}^i)\big\lvert \underline{\pi}_t, x_t^i \right\}.  \label{eq:Vdef}
   }
}
From the equilibrium generating function $\theta$ defined though this backward recursion, the equilibrium strategy and belief profile ($\beta^*,\mu^*$) are defined as follows.

\subsubsection{Forward Recursion}
\begin{itemize} 
\item[(a)] Initialize at time $t=1$, 
\eq{
\mu^{*}_1[\phi](x_1) &:= \prod_{i=1}^N Q^i(x_1^i). \label{eq:mu*def0}
}
\item[(b)] For $t =1,2 \ldots T, \forall i =1,2, u_{1:t}\in \cH_{t+1}^c, x_{1:t}^i \in(\cX^i)^t$
\eq{
\beta_{t}^{*,i}(u_{t}^i|u_{1:t-1}x_{1:t}^i) &:= \theta_{t}^i[\mu_{t}^*[u_{1:t-1}]](u^i_{t}|x_{t}^i) \label{eq:beta*def}
}
and
\seq{
\eq{
\mu^{*,i}_{t+1}[u_{1:t}] &:= \bar{F}(\mu_t^{*,i}[u_{1:t-1}], \theta_t^i[\mu_t^*[u_{1:t-1}]], u_t) \\
\mu^{*}_{t+1}[u_{1:t}](x_t^1,x_t^2)&:= \prod_{i=1}^2\mu^{*,i}_{t+1}[u_{1:t}](x_t^i).  \label{eq:mu*def}
}
}

\end{itemize}

The strategy and belief profile $(\beta^{*}, \mu^*)$ thus constructed form an SPBE of the game~\cite[Theorem 1]{VaAn16}.

\section{SPBE of the dynamic LQG game}
\label{sec:results}

In this section, we apply the general methodology for finding SPBE described in the previous section, on the specific dynamic LQG  game model described in Section~\ref{sec:model}. We show that players' strategies that are linear in their private types in conjunction with Gaussian beliefs, form an SPBE of the game. 
We prove this result by constructing an equilibrium generating function $\theta$ using backward recursion such that for all Gaussian belief vectors $\underline{\pi}_t$, $\tilde{\gamma}_t=\theta_t[\underline{\pi}_t]$,  $\tilde{\gamma}_t^i$ is of the form $\tilde{\gamma}_t^i(u_t^i | x_t^i) =\delta(u_t^i -\btL_t^i x^i - \tm_t^i) $ and satisfies (\ref{eq:M_FP}). Based on $\theta$, we construct an equilibrium belief and strategy profile.

The following lemma shows that common beliefs remain Gaussian under linear deterministic $\gamma_t$ of the form $\gamma_t^i(u_t^i | x_t^i) =\delta(u_t^i - \bL_t^i x_t^i - m_t^i) $.

\begin{lemma}
\label{lemma:belief_update}
If $\pi^i_t$ is a Gaussian distribution with mean $\hat{x}_{t}^i$ and covariance $\bSigma_{t}^i$, and $\gamma_t^i(u_t^i | x_t^i)=\delta(u_t^i - \bL_t^i x_t^i - m_t^i) $ then $\pi_{t+1}^i$, given by (\ref{eq:pi_update}), is also Gaussian distribution with mean $\hat{x}_{t+1}^i$ and covariance $\bSigma_{t+1}^i$, where
\seq{
\eq{
\hat{x}_{t+1}^i &= \bA_t^i\hat{x}_t^i + \bB_t^i u_t +  \bA_t^i\bG_t^i (u_t^i - \bL_t^i\hat{x}_t^i-m_t^i)\label{eq:mean_update}\\
\bSigma_{t+1}^i &= \bA_t^i(\bI - \bG_t^i\bL_t^i)^{\tr}\bSigma_t^i(\bI - \bG_t^i\bL_t^i)\bA_t^{i\tr}+ \bQ^i, \label{eq:cov_update}
}
\label{eq:meancov_update}
}
where 
\eq{\bG_t^i = \bSigma_t^i\bL_t^{i\tr}(\bL_t^i\bSigma_t^i\bL_t^{i\tr})^{-1}. \label{eq:g_def}
}
\end{lemma}
\begin{proof}
See Appendix~\ref{app:belief_update}
\end{proof}

Based on previous lemma, we define $\phi_x^i,\phi_s^i $ as update functions of mean and covariance matrix, respectively, as defined in (\ref{eq:meancov_update}), such that 
\seq{
\label{eq:phi_def}
\eq{
\hat{x}_{t+1}^i &= \phi_x^{i}(\hat{x}_t^i,\bSigma_t^i,\bL_t^i,m_t^i, u_t)\label{eq:phix_def}\\
\bSigma_{t+1}^i &= \phi_s^{i}(\bSigma_t^i,\bL_t^i)\label{eq:phis_def} .
}
}
We also say,
\eq{
\hat{x}_{t+1} &= \phi_x(\hat{x}_t,\bSigma_t,\bL_t,m_t, u_t)\label{eq:phix_def2}\\
\bSigma_{t+1} &= \phi_s(\bSigma_t,\bL_t) \label{eq:phis_def2} .
}
The previous lemma shows that with linear deterministic $\gamma^i_t$, the next update of the mean of the common belief, $\hat{x}_{t+1}^i$ is linear in $\hat{x}_t^i$ and the control action $u_t^i$. Furthermore, these updates are given by appropriate Kalman filter equations. It should be noted however that the covariance update in (\ref{eq:cov_update}) depends on the strategy through $\gamma_t^i$ and specifically through the matrix $\bL_t^i$. This specifically shows how belief updates depend on strategies on the players which leads to signaling, unlike the case in classical stochastic control and the model considered in~\cite{GuNaLaBa14}. 


Now we will construct an equilibrium generating function $\theta$ using the backward recursion in (\ref{eq:VT+1})--(\ref{eq:Vdef}). The $\theta$ function generates linear deterministic partial functions $\gamma_t$, which, from Lemma~\ref{lemma:belief_update} and the fact that initial beliefs (or priors) are Gaussian, generates only Gaussian belief vectors $(\pi_t^1,\pi_t^2)_{t\in\cT}$ for the whole time horizon. These beliefs can be sufficiently described by their mean and covariance processes $(\hat{x}_t^1,\bSigma_t^1)_{t\in\cT}$ and $(\hat{x}_t^2,\bSigma_t^2)_{t\in\cT}$ which are updated using (\ref{eq:meancov_update}). 

For $t = T+1,T,\ldots, 1 $, we define the vectors
\eq{
e_t^i  := \bm{ x_t^i\\ \hat{x}_{t}^1\\ \hat{x}_{t}^2} \quad
 z_t^i  := \bm{u_t^i\\ x_t^i\\ \hat{x}_{t}^1\\ \hat{x}_{t}^2} \quad
 y_t^i := \bm{u_t^1\\u_t^2\\x_t^1\\x_t^2\\x_{t+1}^i\\ \hat{x}_{t+1}^1\\ \hat{x}_{t+1}^2}. \label{eq:ezy_def}
}

\begin{theorem}
\label{thm:back_recr}
The backward recursion (\ref{eq:VT+1})--(\ref{eq:Vdef}) admits\footnote{Under certain conditions, stated in the proof.} a solution of the form $\theta_t[\pi_t] = \theta_t[\hat{x}_t,\bSigma_t] = \tilde{\gamma}_t$ where ${ \tilde{\gamma}^{i}_t}(u_t^{i}|x^{i}_t) = \delta(u_t^i - \btL_t^{i}x_t^{i} - \tm_t^{i})$ and $\btL_t^i, \tm_t^i$ are appropriately defined matrices and vectors, respectively. Furthermore, the value function reduces to
\seq{
\eq{
V_t^i(\underline{\pi}_t,x_t^i) &= V_t^i(\hat{x}_t,\bSigma_t,x_t^i) \\
&= quad(\bV_t^i(\bSigma_t); e_t^i) + \rho_t^i(\bSigma_t).
}
}
with $\bV_t^i(\bSigma_t)$ and $\rho_t^i(\bSigma_t)$ as appropriately defined matrix and scalar quantities, respectively.
\end{theorem}

\begin{proof}
We construct such a $\theta$ function through the backward recursive construction and prove the properties of the corresponding value functions inductively.

\bit{
\item[(a)] For $i=1,2,\forall \ \bSigma_{T+1},$ let $\bV^i_{T+1}(\bSigma_{T+1}) := \bZero, \rho_{T+1}^i(\bSigma_{T+1}) := 0$. 
Then $\forall \ \hat{x}_{T+1}^1, \hat{x}_{T+1}^2$,\\$ \bSigma_{T+1}^1,\bSigma_{T+1}^2, x_{T+1}^i$ and  for $  \underline{\pi}_{t}= (\pi_{t}^1,\pi_{t}^2)$, where $\pi_{t}^i$ is $N(\hat{x}_{t}^i,\bSigma_{t}^i)$,
\seq{
\eq{
V^i_{T+1}(\underline{\pi}_{T+1},x_{T+1}^i) &:=0\\
&= V_{T+1}^i(\hat{x}_{T+1},\bSigma_{T+1},x_{T+1}^i)\\
&=quad(\bV_{T+1}^i(\bSigma_{T+1}),e_{T+1}^i) + \rho_{T+1}^i(\bSigma_{T+1}).
}
}
\item[(b)]  For all $t\in \{T,T-1,\ldots, 1 \}, i=1,2$,

Suppose $V^i_{t+1}(\underline{\pi}_{t+1},x_{t+1}^i)  = quad(\bV_{t+1}^i(\bSigma_{t+1}),e_{t+1}^i) + \rho_{t+1}^i(\bSigma_{t+1})$ (from induction hypothesis) where $\bV^i_{t+1}$ is a symmetric matrix defined recursively. Define $\bVb_t^i$ as 
\eq{
\bVb_t^i(\bSigma_t,\bL_t) &:= \bm{\bT^i&\bS^{i} & \bZero \\ \bS^{i\tr}&\bP^i & \bZero \\  \bZero & \bZero & \bV^i_{t+1}(\phi_s(\bSigma_t,\bL_t)) }. \label{eq:bVb_def}
}
Since $\bT^i,\bP^i$ are symmetric by assumption, $\bVb_t^i$ is also symmetric.

For ease of exposition, we will assume $i=1$ and for player $2$, a similar argument holds. At time $t$, the quantity that is minimized for player $i=1$ in (\ref{eq:M_FP}) can be written as
\eq{
&\E^{\gamma^1_t(\cdot|x_t^1)}\left[ \E^{ \tilde{\gamma}^{2}_t} \left[ R^1(X_t,U_t) +  V_{t+1}^1 (F(\underline{\pi}_t, \tilde{\gamma}_t, U_t), X_{t+1}^1) \big\lvert  \underline{\pi}_t ,x_t^1,u_t^1 \right] \big\lvert  \underline{\pi}_t ,x_t^1\right].
 }
The inner expectation can be written as follows, where ${ \tilde{\gamma}^{2}_t}(u_t^{2}|x^{2}_t) =\delta(u_t^2-  \btL_t^{2}x_t^{2} - \tm_t^{2})$,
\seq{
\eq{
&\E^{ \tilde{\gamma}^{2}_t} \left[quad\left(\bm{\bT^1&\bS^{1}\\\bS^{1\tr}&\bP^1};z_t^i\right) + quad\left(\bV_{t+1}^1(\phi_s(\bSigma_{t},\btL_t));e_{t+1}^i\right) + \rho_{t+1}^1(\phi_s(\bSigma_{t},\btL_t)) \big\lvert \pi_t,x_t^1, u_t^i\right]\\
 &= \E^{ \tilde{\gamma}^{2}_t} \left[ quad\left(\bVb_t^1(\bSigma_{t},\btL_t);  y_t^1\right) + \rho_{t+1}^1(\phi_s(\bSigma_{t},\btL_t))  \big\lvert \pi_t,x_t^1,  u_t^1\right]  \\ 
 %
 %
 &= quad\left(\bVb_t^1(\bSigma_{t},\btL_t); \bD_t^1z_t^1+ \bC_t^1 \bm{m_t^1\\ \tm_t^2}  \right) + \rho_{t}^1(\bSigma_t), \label{eq:term2min}
 }
 }
where $\bVb_t^i$ is defined in (\ref{eq:bVb_def}) and function $\phi_s$ is defined in (\ref{eq:phis_def2}); $y_t^i, z_t^i $ are defined in (\ref{eq:ezy_def}); $\rho_t^i$ is given by 
\eq{
\rho_t^i(\bSigma_t) &= tr\left(\bSigma_t^{-i} quad\left(\bVb_t^i(\bSigma_t,\btL_t); \bJ_t^i \right)\right) +tr(\bQ^iV_{11,t+1}^i(\phi_s(\bSigma_t,\btL_t))) + \rho_{t+1}^i(\phi_s(\bSigma_t,\btL_t)), \label{eq:rho_update}
}
where $V_{11,t+1}^i$ is the matrix corresponding to $x_{t+1}^i$ in $V_{t+1}^i$ i.e. in the first row and first column of the matrix $V_{t+1}^i$; and matrices $\bD_t^i, \bC_t^i, \bJ_t^i$ are as follows,

\seq{
\eq{
 &\bD_t^1:= \bm{\bI &\bZero &\bZero &\bZero \\\bZero & \bZero&\bZero &\btL_t^2  \\\bZero&\bI&\bZero&\bZero\\ \bZero&\bZero&\bZero & \bI\\\bB^1_{1,t} &\bA_t^1 & \bZero&\bB_{2,t}^1\btL_t^2 \\\bA_t^1\bG_t^1+\bB_{1,t}^1 &\bZero &\bA_t^1(\bI - \bG_t^1\bL_t^1) & \bB_{2,t}^1\btL_t^2\\\bB_{1,t}^2 &\bZero &\bZero &\bA_t^2  + \bB_{2,t}^2\btL_t^2 }
 }
 \eq{
& \bD_t^2:= \bm{\bZero &\bZero &\btL_t^1 &\bZero \\\bI & \bZero&\bZero & \bZero \\ \bZero&\bZero&\bI&\bZero\\ \bZero&\bI &\bZero & \bZero\\\bB^2_{2,t} &\bA_t^2 &\bB_{1,t}^2\btL_t^1&\bZero \\
 \bB_{2,t}^1 &\bZero &\bA_t^1+ \bB_{1,t}^1\btL_t^1&\bZero \\\bA_t^2\bG_t^2+\bB_{2,t}^2 &\bZero& \bB_{1,t}^2\btL_t^1 &\bA_t^2 (\bI- \bG_t^2\bL_t^2) } 
 }
 }
 \eq{
 \bC_t^1&:= \bm{\bZero & \bZero \\\bZero & \bI \\\bZero & \bZero\\\bZero & \bZero \\ \bZero & \bB_{2,t}^1 \\-\bA_t^1\bG_t^1&  \bB_{2,t}^1 \\\bZero & \bB_{2,t}^2 } \quad
 \bC_t^2 := \bm{\bI & \bZero \\ \bZero & \bZero\\\bZero & \bZero\\\bZero & \bZero \\\bB_{1,t}^2& \bZero \\\bB_{1,t}^1 & \bZero \\ \bB_{1,t}^2& -\bA_t^2\bG_t^2 }\label{eq:DC_def}
 }
 \eq{
& \bJ_t^{1\tr}:= \bm{\bZero & \bL_t^2 & \bZero & \bI & \bB_{2,t}^1\bL_t^2 & \bB_{2,t}^1\bL_t^2 & (\bB_{2,t}^2+ \bA_t^2\bG_t^2)\bL_t^2}^{\tr}\nonumber \\
& \bJ_t^{2\tr}:=\bm{\bL_t^1& \bZero& \bI & \bZero & \bB_{1,t}^2\bL_t^1 & (\bB_{1,t}^1 + \bA_t^1\bG_t^1)\bL_t^1 & \bB_{1,t}^2\bL_t^1}^{\tr}\label{eq:J_def}
}
where $\bB_t^i =: \bm{\bB_{1,t}^i & \bB_{2,t}^i}$, $\bB_{1,t}^i, \bB_{2,t}^i$ are the parts of the matrix $\bB_{t}^i$ that corresponds to $u_t^1,u_t^2$ respectively. 
Let $\bD_t^1 =:\bm{\bD_t^{u1} &\bD_t^{e1}}$ where $\bD_t^{u1}$ is the first column matrix of $\bD_t^1$ corresponding to $u_t^1$ and $\bD_t^{e1}$ is the matrix composed of remaining three column matrices of $\bD_t^1$ corresponding to $e_t^1$.
The expression in (\ref{eq:term2min}) is averaged with respect to $u_t^1$ using the measure $\gamma_t^1(\cdot|x_t^1)$ and minimized in (\ref{eq:M_FP}) over $\gamma_t^1(\cdot|x_t^1)$. This minimization can be performed component wise leading to a deterministic policy $\tilde{\gamma}_t^1(u_t^1|x_t^1) = \delta(u_t^1 - \btL_t^1x_t^1 - \tm_t^1)=\delta(u_t^1 - u_t^{1*})$,
assuming that the matrix $\btD_t^{u1\tr}\bVb_t^{1}\btD_t^{u1}$ is positive definite\footnote{This condition is true if the instantaneous cost matrix $\bR^i =\bm{\bT^i &\bS^i \\ \bS^{i\tr} & \bP^i}$ is positive definite and can be proved inductively in the proof by showing that $\bV_t^{i}$ and $\bVb_t^{i}$ are positive definite.}. In that case, the unique minimizer $u_t^{1*} = \btL_t^1 x_t^1 + \tm_t^1$ can be found by differentiating (\ref{eq:term2min}) w.r.t. $u_t^{1\tr}$ and equating it to $\bZero$, resulting in the equation,
\seq{
\eq{
\bZero&=2\bm{\bI&\bZero&\bZero&\bZero}\btD_t^{1\tr}\bVb_t^{1}(\bSigma_t,\btL_t)\left(\btD_t^1z_t^1+\btC_t^1\tm_t \right)\\
 \bZero&=\btD_t^{u1\tr}\bVb_t^{1}(\bSigma_t,\btL_t)\left(\btD_t^{u1} u_t^{1*} + \btD_t^{e1}e_t^1 + \btC_t^1\tm_t\right)\\
 \bZero&=\btD_t^{u1\tr}\bVb_t^{1}(\bSigma_t,\btL_t)\left(\btD_t^{u1} (\btL_t^1 x_t^1 + \tm_t^1) + [\btD_t^{e1}]_1x_t^1 + [\btD_t^{e1}]_{23}\hat{x}_t + \btC_t^1\tm_t\right),\label{eq:M_FP_eqv}
}
}
where $[\bD^{ei}]_{1}$ is the first matrix column of $\bD^{ei}$, $[\bD^{ei}]_{23}$ is the matrix composed of the second and third column matrices of  $\bD^{ei}$. Matrices $\btD_t^i, \btC_t^i$ are obtained by substituting $\btL_t^i,\btG_t^i$ in place of $\bL_t^i,\bG_t^i$ in the definition of $\btD_t^i, \btC_t^i$ in (\ref{eq:DC_def}), respectively, and $\btG_t^i$ is the matrix obtained by substituting $\btL_t^i$ in place of $\bL_t^i$ in (\ref{eq:g_def}). 

Thus (\ref{eq:M_FP_eqv}) is equivalent to (\ref{eq:M_FP}) and with a similar analysis for player 2, it implies that $\btL_t^i$ is solution of the following algebraic fixed point equation, 
\seq{
\label{eq:L_fp}
\eq{
&\left(\btD_t^{ui\tr}\bVb_t^{i}(\bSigma_t,\btL_t)\btD_t^{ui} \right)\btL_t^i = -\btD_t^{ui\tr}\bVb_t^{i}(\bSigma_t,\btL_t) [\btD_t^{ei}]_1. 
}
For player 1, it reduces to,
\eq{
\label{eq:L_fp2}
&\hspace{-10pt}\left[\bT_{11}^{1} + \bm{\bB_{1,t}^1 \\  \bA_t^1\bG_t^1+ \bB_{1,t}^1\\  \bB_{1,t}^2 }^{\tr}\bV_{t+1}^{1}(\phi_s(\bSigma_t,\btL_t)) \bm{ \bB_{1,t}^1  \\ \bA_t^1\bG_t^1+ \bB_{1,t}^1  \\\bB_{1,t}^2  } \right] \btL_t^1 \nonumber \\
&=  - \left[ \bS_{11}^{1\tr} +  \bm{\bB_{1,t}^1 \\  \bA_t^1\bG_t^1+ \bB_{1,t}^1\\  \bB_{1,t}^2 }^{\tr}\bV_{t+1}^{1}(\phi_s(\bSigma_t,\btL_t)) \bm{\bA_t^1\\0\\0}\right],
}
}
and a similar expression holds for player 2.

In addition, $\tm_t$ can be found from (\ref{eq:M_FP_eqv}) as
\seq{
\eq{
&\bm{\btD_t^{u1\tr}\bVb_t^{1}\btD_t^{u1}  & \bZero \\ \bZero &\btD_t^{u2\tr}\bVb_t^{2}\btD_t^{u2}  }\tm_t 
= -\bm{\btD_t^{u1\tr}\bVb_t^{1}[\btD_t^{e1} ]_{23}\\  \btD_t^{u2\tr}\bVb_t^{2}[ \btD_t^{e2}]_{23}} \hat{x}_t- 
\bm{\btD_t^{u1\tr}\bVb_t^{1} \btC_t^1\\\btD_t^{u2\tr}\bVb_t^{2} \btC_t^2} \tm_t \\
 \tm_t &= -\left[ \bm{\btD_t^{u1\tr}\bVb_t^{1}\btD_t^{u1}  & \bZero \\ \bZero &\btD_t^{u2\tr}\bVb_t^{2}\btD_t^{u2}  } +\bm{\btD_t^{u1\tr}\bVb_t^{1} \btC_t^1\\  \btD_t^{u2\tr}\bVb_t^{2} \btC_t^2} \right]^{-1}
 \bm{\btD_t^{u1\tr}\bVb_t^{1}[\btD_t^{e1}]_{23}\\  \btD_t^{u2\tr}\bVb_t^{2} [\btD_t^{e2}]_{23}} \hat{x}_t\\
&=:\btM_t \hat{x}_t =: \bm{\btM_t^1\\\btM_t^2} \hat{x}_t, \label{eq:m_eq}
}
}

 Finally, the resulting cost for player $i$ is,
\seq{
\eq{
V^i_{t}(\underline{\pi}_t,x_t^i)  &= V^i_{t}(\hat{x}_t,\bSigma_t,x_t^i) \\
&:=quad\left(\bVb_t^i(\bSigma_t,\btL_t); \bm{\btD_t^{ui} & \btD_t^{ei}} \bm{ \btL_t^i x_t^i + \btM_t^i\hat{x}_t\\ e_t^i }+ \btC_t^i\btM_t\hat{x}_t  \right)  + \rho_t^i(\bSigma_t)\\
&=quad\left(\bVb_t^i(\bSigma_t,\btL_t); \btD_t^{ui}(\btL_t^i x_t^i + \btM_t^i\hat{x}_t )+ \btD_t^{e1}e_t^i + \btC_t^i \btM_t \hat{x}_t \right) + \rho_t^i(\bSigma_t)\\
&=quad\left(\bVb_t^i(\bSigma_t,\btL_t); \left(\bm{\btD_t^{ui}\btL_t^i  & \btD_t^{ui}\btM_t^i+ \btC_t^i\btM_t}  + \btD_t^{ei}\right)e_t^i \right)  + \rho_t^i(\bSigma_t)\\
&=quad\left(\bVb_t^i(\bSigma_t,\btL_t);\btF_t^ie_t^i \right) + \rho_t^i(\bSigma_t)\\
&=quad\left(\btF_t^{i\tr}\bVb_t^i(\bSigma_t,\btL_t)\btF_t^i; e_t^i \right) + \rho_t^i(\bSigma_t)\\
&=quad\left(\bV_t^i(\bSigma_{t}); e_t^i \right) + \rho_t^i(\bSigma_t),
}
}
where,
\seq{
\label{eq:V_update}
\eq{
\btF_t^i &:=\bm{\btD_t^{ui}\btL_t^i  & \btD_t^{ui}\btM_t^i+ \btC_t^i\btM_t} + \btD_t^{ei}\\
\bV_t^i(\bSigma_{t}) &:= \btF_t^{i\tr}\bVb_t^i(\bSigma_t,\btL_t)\btF_t^i. \label{eq:bV_def}
}
}
Since $\bVb_t^i$ is symmetric, so is $\bV_t^i$. Thus the induction step is completed.
}
\end{proof}

Taking motivation from the previous theorem and with slight abuse of notation, we define 
\eq{
\tilde{\gamma}_t &= \theta_t[\underline{\pi}_t] = \theta_t[\hat{x}_t,\bSigma_t],
}
and since $\tilde{\gamma}_t^i(u_t^i|x_t^i) = \delta(u_t^i- \btL_t^i x_t^i - \tm_t^i)$, we define a reduced mapping $(\theta^L,\theta^m)$ as
\eq{
 \theta_t^{Li}[\hat{x}_t,\bSigma_t] = \theta_t^{Li}[\bSigma_t] := \btL_t^i \;\; \text{ and } \;\; \theta_t^{mi}[\hat{x}_t,\bSigma_t] := \tm_t^i, 
}
where $\btL_t^i$ does not depend on $\hat{x}_t$ and $\tm_t^i$ is linear in $\hat{x}_t$ and is of the form $\tm_t^i = \btM_t^i\hat{x}_t$.

Now we construct the equilibrium strategy and belief profile $(\beta^*,\mu^*)$ through the forward recursion in (\ref{eq:mu*def0})--(\ref{eq:mu*def}), using the equilibrium generating function $\theta \equiv (\theta^L,\theta^m)$.
\bit{
\item[(a)] Let \eq{
\mu^{*,i}_1[\phi](x_1^i)= N(0,\bSigma_1^i).
} 

\item[(b)] For $t =1,2 \ldots T-1, \forall u_{1:t} \in \cH_{t+1}^c$, if $\mu_t^{*,i}[u_{1:t-1}] = N(\hat{x}_{t}^i, \bSigma_t^i)$, let $\btL_t^i =  \theta_t^{Li}[\bSigma_t], \tm_t^i =  \theta_t^{mi}[\hat{x}_t,\bSigma_t] = \btM_t^i\hat{x}$. Then $ \forall  x_{1:t}^i \in(\cX^i)^t$
\seq{
\label{eq:meancov_update_FI}
\eq{
\beta_{t}^{*,i}(u_{t}^i|u_{1:t-1}x_{1:t}^i) &:= \delta(u_t^i - \btL_t^ix_t^i - \btM_t^i\hat{x}_t)\\
\mu_{t+1}^{*,i}[u_{1:t}] &:= N(\hat{x}_{t+1}^i ,\bSigma_{t+1}^i) \\
\mu^{*}_{t+1}[u_{1:t}](x_t^1,x_t^2)&:= \prod_{i=1}^2\mu^{*,i}_{t+1}[u_{1:t}](x_t^i),
} 
}

where $\hat{x}_{t+1}^i = \phi_x^{i}(\hat{x}_t^i, \btL_t^i, \tm_t^i, u_t)$ and $\bSigma_{t+1}^i = \phi_s^{i}(\bSigma_t^i,\btL_t^i)$.
}

\begin{theorem}
 $(\beta^*,\mu^*)$ constructed above is a PBE of the dynamic LQG game.
\end{theorem}
\begin{proof}
The strategy and belief profile $(\beta^*,\mu^*)$ is constructed using the forward recursion steps (\ref{eq:mu*def0})--(\ref{eq:mu*def}) on equilibrium generating function $\theta$, which is defined through backward recursion steps (\ref{eq:VT+1})--(\ref{eq:Vdef}) implemented in the proof Theorem~\ref{thm:back_recr}. Thus the result is directly implied by Theorem 1 in~\cite{VaAn16}.
\end{proof}

\section{Disussion}
\label{sec:disc}
\subsection{Existence}
In the proof of Theorem~\ref{thm:back_recr}, $\btD_t^{u1\tr}\bVb_t^{1}\btD_t^{u1}$ is assumed to be positive definite. This can be achieved if $\bR^i$ is positive definite, through which it can be easily shown inductively in the proof of Theorem~\ref{thm:back_recr} that the matrices $\bV_t^1,\bVb_t^1$ are also positive definite.

Constructing the equilibrium generating function $\theta$ involves solving the algebraic fixed point equation in (\ref{eq:L_fp}) for $\btL_t$ for all $\bSigma_t$. In general, the existence is not guaranteed, as is the case for existence of $\tilde{\gamma}_t$ in (\ref{eq:M_FP}) for general dynamic games with asymmetric information. At this point, we don't have a general proof for existence. However, in the following lemma, we provide sufficient conditions on the matrices $\bA_t^i, \bB_t^i, \bT^i, \bS^i, \bP^i,\bV_{t+1}^i$ and for the case $m^i=1$, for a solution to exist.  

\begin{lemma}
\label{lemma:exist}
For $m^1=m^2=1$, there exists a solution to (\ref{eq:L_fp}) if and only if for $i=1,2$, $\exists\ l^i\ \in \R^{n^i}$ such that $l^{i\tr}\bDelta^i(l^1,l^2) l^i\geq 0$, or sufficiently $\bDelta^i(l^1,l^2) + \bDelta^{i,\tr}(l^1,l^2)$ is positive definite, where $\bDelta^i, i=1,2$ are defined in Appendix~\ref{app:exist_proof}.
\end{lemma}
\begin{proof}
See Appendix~\ref{app:exist_proof}.
\end{proof}

\subsection{Steady state}
In Section~\ref{sec:prelim}, we presented the backward/forward methodology to find SPBE for finite time-horizon dynamic games, and specialized that methodology in this chapter, in Section~\ref{sec:results}, to find SPBE for dynamic LQG games with asymmetric information, where equilibrium strategies are linear in players' types. It requires further investigation to find the conditions for which the backward-forward methodology could be extended to infinite time-horizon dynamic games, with either expected discounted or time-average cost criteria. Such a methodology for infinite time-horizon could be useful to characterize steady state behavior of the games. Specifically, for time homogenous dynamic LQG games with asymmetric information (where matrices $\bA^i, \bB^i$ are time independent), under the required technical conditions for which such a methodology is applicable, the steady state behavior can be characterized by the fixed point equation in matrices $(\bL^i, \bSigma^i ,\bV^i)_{i=1,2}$ through \eqref{eq:phis_def2}, (\ref{eq:L_fp2}) and (\ref{eq:V_update}), where the time index is dropped in these equations, i.e. for $i=1,2,$

\bit{

\item[1.] \hfill \makebox[0pt][r]{%
            \begin{minipage}[b]{\textwidth}
              \begin{equation}
                 \bSigma = \phi_s(\bSigma,\bL) \label{eq:SS1}
              \end{equation}
          \end{minipage}}
          
\item[2.] \hfill \makebox[0pt][r]{%
            \begin{minipage}[b]{\textwidth}
              \begin{equation}
                \left(\bD^{ui\tr}\bVb^{i}\bD^{ui} \right)\bL^i = -\bD^{ui\tr}\bVb^{i} [\bD^{ei}]_1 \label{eq:SS2}
              \end{equation}
          \end{minipage}}
       
\item[3.] \hfill \makebox[0pt][r]{%
            \begin{minipage}[b]{\textwidth}
              \begin{equation}
                \bV^i = \bF^{i\tr}\bVb^i\bF^i,\label{eq:SS3}
              \end{equation}
          \end{minipage}}                    
%
%
}

where $\bVb^i = \bm{\bT^i&\bS^{i} & \bZero \\ \bS^{i\tr}&\bP^i & \bZero \\  \bZero & \bZero & \bV^i }$.

Observe that in the above equations the matrices $\bV^i$ and $\bVb^i$ do not appear as functions of $\bSigma$, as in the finite horizon case described in \eqref{eq:bVb_def}, \eqref{eq:bV_def}, in the proof of Theorem~\ref{thm:back_recr}. The reason for that is as follows.
The steady state behavior for a general dynamic game with asymmetric information and independent types, if it exists, would involve fixed point equation in value functions $(V^i(\cdot))_i$. However, for the LQG case, it reduces to a fixed point equation in $(V^i(\bSigma))_i$, i.e. value functions evaluated at a specific value of $\bSigma$. This is so because the functions $V^i$ are evaluated at $\bSigma$ and $\phi(\bSigma,\bL)$, which at steady state are exactly the same (see \eqref{eq:SS1}).
As a result, the fixed point equation reduces to the three algebraic equations as shown above with variables the matrices $\bSigma$, $\bL$, $\bVb$ and $\bV$, which represents an enormous reduction in complexity.

\subsubsection{Numerical examples}
In this section, we present numerically found solutions for steady state, assuming that our methodology extends to the infinite horizon problem for the model considered. We assume $\bB^i=0$ which implies that the state process $(X_t^i)_{t\in \cT}$ is uncontrolled.
\bit{
\item[1.] For $i=1,2$, $ m^i=1, n^i=2, \bA^i = 0.9\bI, \bB^i=\bZero, \bQ^i = \bI$, 
\eq{
\bT^1 = \bm{\bI & \frac{1}{4}\bI\\\frac{1}{4}\bI &\bZero }, \quad 
\bT^2 = \bm{\bZero & \frac{1}{4}\bI\\\frac{1}{4}\bI &\bI },  \quad  \bP^1 = \bm{\bI  &\bZero  \\ \bZero & \bZero }, \nonumber \\
\bP^2 = \bm{\bZero  &\bZero  \\ \bZero & \bI }, \quad
\bS^1 = \bm{ \bOne   &\bZero  \\ \bZero & \bZero  }, \quad \bS^2 = \bm{ \bZero &\bZero  \\\bZero  & \bOne },
}
there exists a symmetric solution as, for $i=1,2,$
\eq{
\bL^i = -\bm{1.062 & 1.062} , \bSigma^i = \bm{3.132 & -2.132 \\ -2.132 & 3.132}. } 

\item[2.] For $i=1,2$, $ m^i=2, n^i=2, \bA^1 = \bm{0.9 & 0 \\ 0 & 0.8}, \bA^2 =0.9\bI, $ and $\bB^i,\bT^i,\bP^i,\bS^i$ used as before with appropriate dimensions, there exists a solution,

\eq{
\bL^1 =-\bm{ 1.680 &1.600 \\ 0.191 &0.286 },\quad \bL^2 = -\bm{ 1.363 &1.363 \\ 1.363 &1.363 } \nonumber \\
 \hspace{-1cm}\bSigma^1 =\bI, \; \quad  
\bSigma^2 = \bm{ 3.132 &-2.132 \\ -2.132 &3.132 }.
 } 
It is interesting to note that for player 1, where $\bA^1$ does not weigh the two components equally, the corresponding $\bL^1$ is full rank, and thus reveals her complete private information. Whereas for player 2, where $\bA^2$ has equal weight components, the corresponding $\bL^2$ is rank deficient, which implies, at equilibrium player 2 does not completely reveal her private information. Also it is easy to check from (\ref{eq:cov_update}) that with full rank $\bL^i$ matrices, steady state $\bSigma^i  = \bQ^i$. 
}

\section{Conclusion}
\label{sec:concl}

In this paper, we study a two-player dynamic LQG game with asymmetric information and perfect recall where players' private types evolve as independent controlled Markov processes. We show that under certain conditions, there exist strategies that are linear in players' private types which, together with Gaussian beliefs, form a PBE of the game. We show this by specializing the general methodology developed in~\cite{VaAn16} to our model. Specifically, we prove that (a) the common beliefs remain Gaussian under the strategies that are linear in players' types where we find update equations for the corresponding mean and covariance processes; (b) using the backward recursive approach of~\cite{VaAn16}, we compute an equilibrium generating function $\theta$ by solving a fixed point equation in linear deterministic partial strategies $\gamma_t$ for all possible common beliefs and all time epochs. Solving this fixed point equation reduces to solving a matrix algebraic equation for each realization of the state estimate covariance matrices. Also, the cost-to-go value functions are shown to be quadratic in private type and state estimates. This result is one of the very few results available on finding signaling perfect Bayesian equilibria of a truly dynamic game with asymmetric information.

	\appendices
	\section{}
	\label{app:belief_update}
	This lemma could be interpreted as Theorem~2.30 in \cite[Ch. 7]{KuVa86} with appropriate matrix substitution where specifically, their observation matrix $C_k$ should be substituted by our $L_k$. We provide an alternate proof here for convenience. 
	
	$\pi_{t+1}^i$ is updated from $\pi_t^i$ through (\ref{eq:pi_update}).
Since $\pi_t^i$ is Gaussian, $\gamma_t^i(u_t^i|x_t^i) = \delta(u_t^i - L_t^ix_t^i - m_t^i)$ is a linear deterministic constraint and kernel $Q^i$ is Gaussian, thus $\pi_{t+1}^i $ is also Gaussian. We find its mean and covariance as follows.
	
	 We know that $x_{t+1}^i=\bA_t^i x_t^i + \bB_t^i u_t + w_t^i$. Then, 
	\seq{
	\eq{
&\E[X_{t+1}^i| \pi^i_t,\gamma_t^i, u_t] \nonumber \\
&= \E[\bA_t^i X_t^i + \bB_t^i U_t + W_t^i| \pi^i_t,\gamma_t^i, u_t] \\
&=   \bA_t^i\E[X_t^i | \pi^i_t,\gamma_t^i, u_t]  + \bB_t^i u_t\label{eq:a1_w0}\\
&=  \bA_t^i \E[X_t^i | \bL_t^iX_t^i = u_t^i-m_t^i ]  + \bB_t^i u_t \label{eq:a1_ci}
}
}
where (\ref{eq:a1_w0}) follows because $W_t^i$ has mean zero. 
Suppose there exists a matrix $\bG_t^i$ such that $X_t^i-\bG_t^i\bL_t^iX_t^i$ and $\bL_t^iX_t^i$ are independent.  
 Then
 \seq{
\eq{
&\E[X_t^i\big|\bL_t^iX_t^i = u_t^i-m_t^i] \\
&= \E[X_t^i - \bG_t^i\bL_t^iX_t^i + \bG_t^i\bL_t^iX_t^i \big| \bL_t^iX_t^i = u_t^i-m_t^i] \\
&= \E[X_t^i-\bG_t^i\bL_t^iX_t^i] + \bG_t^i(u_t^i-m_t^i)\\
&= \hat{x}_t^i + \bG_t^i(u_t^i- \bL_t^i\hat{x}_t^i - m_t^i),
}
}
where $\bG_t^i$ satisfies
\seq{
\eq{
\E[(X_t^i-\bG_t^i\bL_t^iX_t^i)&(\bL_t^iX_t^i)^{\tr}]\nonumber \\
&\hspace{-1cm}= \E[(X_t^i-\bG_t^i\bL_t^iX_t^i)]\E[(\bL_t^iX_t^i)^{\tr}]\\
(\bI-\bG_t^i\bL_t^i)\E[X_t^iX_t^{i\tr}]\bL_t^{i\tr} &= (\bI-\bG_t^i\bL_t^i)\E[X_t^i]\E[X_t^{i\tr}]\bL_t^{i\tr}\\
(\bI-\bG_t^i\bL_t^i)(\bSigma_t^i + \hat{x}_t^i\hat{x}_t^{i\tr})\bL_t^{i\tr} &= (\bI-\bG_t^i\bL_t^i) \hat{x}_t^i\hat{x}_t^{i\tr} \bL_t^{i\tr}\\
\bG_t^i & = \bSigma_t^i \bL_t^{i\tr}(\bL_t^i\bSigma_t^i \bL_t^{i\tr})^{-1}.
}
}
\seq{
\eq{
&\bSigma_{t+1}^i = sm\left(\bA_t^iX_t^i-\E[\bA_t^iX_t^i|\bL_t^iX_t^i = u_t^i-m_t^i] |\bL_t^iX_t^i = u_t^i-m_t^i\right) + \bQ^i
}
}
Now
\seq{
\eq{
&sm\left(X_t^i-\E[X_t^i|\bL_t^iX_t^i = u_t^i-m_t^i] |\bL_t^iX_t^i = u_t^i-m_t^i\right) \\
&= sm\left((X_t^i-\bG_t^i\bL_t^iX_t^i)-(\E[X_t^i-\bG_t^i\bL_t^iX_t^i|\bL_t^iX_t^i = u_t^i-m_t^i] ) | \bL_t^iX_t^i = u_t^i-m_t^i\right)\\
&= sm\left((X_t^i-\bG_t^i\bL_t^iX_t^i) -(\E[X_t^i - \bG_t^i\bL_t^iX_t^i] )\right) \\
&= sm\left((\bI-\bG_t^i\bL_t^i)(X_t^i-\E[X_t^i])\right) \\
&= (\bI-\bG_t^i\bL_t^i)\bSigma_t^i(\bI-\bG_t^i\bL_t^i)^{\tr}
}
}
\section{}
\label{app:exist_proof}
We prove the lemma for player 1 and the result follows for player 2 by similar arguments. For the scope of this appendix, we define $\bbB_t^1 = \bm{\bB_{1,t}^1 \\  \bB_{1,t}^1\\  \bB_{1,t}^2 }$ and for any matrix $\bV$, we define $\bV_{*i}, \bV_{i*}$ as the $i^{th}$ column and the $i^{th}$ row of $\bV$, respectively. Then the fixed point equation (\ref{eq:L_fp}) can be written as,
\eq{
&0 =\left[\bT_{11}^{1} + (\bA_t^1\bG_t^1)^{\tr}\bV_{22,t+1}^{1} ( \bA_t^1\bG_t^1) + \right. \nonumber \\
&\left.\hspace{-5pt} \bbB_t^{1\tr} \bV_{*2,t+1}^1\bA_t^1\bG_t^1 + 
(\bA_t^1\bG_t^1)^{\tr}\bV_{2*,t+1}^1\bbB_t^1 +  \bbB_t^{1\tr} \bV_{t+1}^1\bbB_t^1\right] \bL_t^1  \nonumber \\
&+ \left[ \bS_{11}^{1\tr} +   (\bA_t^1\bG_t^1)^{\tr}\bV_{21,t+1}^{1} \bA_t^1 + \bbB_t^{1\tr} \bV_{*1,t+1}^1\bA_t^1\right].
\label{eq:app3_48}
}
It should be noted that $\bV^i_{t+1}$ is a function of $\bSigma_{t+1}$, which is updated through $\bSigma_t$ and $\bL_t$ as $\bSigma_{t+1} = \phi_s(\bSigma_{t},\bL_t)$ (we drop this dependence here for ease of exposition).
Substituting $\bG_t^1 = \bSigma_t^1\bL_t^{1\tr}(\bL_t^1\bSigma_t^1\bL_t^{1\tr})^{-1}$ and multiplying (\ref{eq:app3_48}) by $(\bL_t^1\bSigma_t^1\bL_t^{1\tr})$ from left and $(\bSigma_t^1\bL_t^{1\tr})$ from right, we get

\eq{
%
%
0&=\bL_t^1\bSigma_t^1\left[ \bL_t^{1\tr}(\bT_{11}^{1} + \bbB_t^{1\tr} \bV_{t+1}^1\bbB_t^1  ) \bL_t^{1}  +  \bA_t^{1\tr}\bV_{22,t+1}^{1} \bA_t^1\right.  \nonumber \\
&   + \bL_t^{1\tr}(\bbB_t^{1\tr} \bV_{*2,t+1}^1\bA_t^1+ \bS_{11}^{1\tr} + \bbB_t^{1\tr} \bV_{*1,t+1}^1\bA_t^1 )  \nonumber \\
 &\left.+ (\bA_t^{1\tr}\bV_{2*,t+1}^1\bbB_t^1+ \bA_t^{1\tr}\bV_{21,t+1}^{1} \bA_t^1)\bL_t^{1} \right] \bSigma_t^1\bL_t^{1\tr}
}

Let $\bbL_t^i = \bL_t^i (\bSigma_t^i)^{1/2}, \bbA_t^i = \bA_t^i(\bSigma_t^i)^{1/2}$,
\seq{
\label{eq:Lambda_def}
\eq{
\bLambda^1_a(\bL_t) & :=\bT_{11}^{1} + \bbB_t^{1\tr} \bV_{t+1}^1\bbB_t^1 \\
 \bLambda^1_b(\bL_t) &:=\bbA_t^{1\tr}\bV_{22,t+1}^{1} \bbA_t^1 \\
\bLambda^1_c(\bL_t) & := \bbB_t^{1\tr} \bV_{*2,t+1}^1\bbA_t^1+ \bS_{11}^{1\tr}(\bSigma_t^1)^{1/2} + \bbB_t^{1\tr} \bV_{*1,t+1}^1\bbA_t^1 \\
  \bLambda^1_d(\bL_t) &:= \bbA_t^{1\tr}\bV_{2*,t+1}^1\bbB_t^1+ \bbA_t^{1\tr}\bV_{21,t+1}^{1} \bbA_t^1.
}
}
Then, 
\eq{
0&= \bbL_t^1\bbL_t^{1\tr}\bLambda^1_a(\bL_t) \bbL_t^{1}\bbL_t^{1\tr}  +  \bbL_t^1\bLambda^1_b(\bL_t) \bbL_t^{1\tr}   
  + \bbL_t^1\bbL_t^{1\tr}\bLambda^1_c(\bL_t) \bbL_t^{1\tr} +  \bbL_t^1\bLambda^1_d(\bL_t) \bbL_t^{1} \bbL_t^{1\tr} \label{eq:Lmatquad}
}
Since m=1, $\Lambda^1_a$ is a scalar. Let $\bbL_t^i=\lambda^i l^{i\tr}$, where $\lambda^i= || \bbL_t^i ||_2$ and $l^i$ is a normalized vector and $t^1 =T_{11}$. Moreover, since the update of $\bSigma_t$ in \eqref{eq:cov_update} is scaling invariant, $\bV^1_{t+1}$ only depends on the directions $l= (l^1,l^2)$. Then, (\ref{eq:Lmatquad}) reduces to the following quadratic equation in $\lambda^1$
\eq{
(\lambda^1)^2 \Lambda^1_a(l)  + \lambda^1 (\bLambda^1_c(l) l^1 +  l^{1\tr}\bLambda^1_d(l)) +  l^{1\tr}\bLambda^1_b(l) l^1   = 0.
}
There exists a real-valued solution\footnote{Note that a negative sign of $\lambda^1$ can be absorbed in $l^1$.} of this quadratic equation in $\lambda^1$ if and only if
\seq{
\eq{
(\bLambda_c(l) l^1 +  l^{1\tr}\bLambda^1_d(l))^2 \geq 4 \Lambda^1_a(l) l^{1\tr}\bLambda^1_b(l) l^1  \\
l^{1\tr}(\bLambda^{1\tr}_c (l)\bLambda^1_c(l)  + \bLambda^1_d(l) \bLambda^{1\tr}_d(l) + 2\bLambda^1_d(l)\bLambda^1_c(l) - 4 \Lambda^1_a(l) \bLambda^1_b(l)) l^1 \geq 0. \label{eq:suff_cond}
}
}
\eq{ 
\text{Let } \bDelta^1(l):= (\bLambda^{1\tr}_c(l) \bLambda^1_c(l)  + \bLambda^1_d(l) \bLambda^{1\tr}_d(l) + 2\bLambda^1_d(l)\bLambda^1_c(l) - 4 \Lambda^1_a(l) \bLambda^1_b(l)).} 
There exists a solution to the fixed point equation (\ref{eq:L_fp}) if and only if $\exists l^1,l^2\in \R^n$ such that $l^{1\tr}\bDelta^1(l) l^1\geq 0$, or sufficiently $\bDelta^1(l) + \bDelta^{1\tr}(l)$ is positive definite.

\bibliographystyle{IEEEtran}

\begin{thebibliography}{10}
\providecommand{\url}[1]{#1}
\csname url@samestyle\endcsname
\providecommand{\newblock}{\relax}
\providecommand{\bibinfo}[2]{#2}
\providecommand{\BIBentrySTDinterwordspacing}{\spaceskip=0pt\relax}
\providecommand{\BIBentryALTinterwordstretchfactor}{4}
\providecommand{\BIBentryALTinterwordspacing}{\spaceskip=\fontdimen2\font plus
\BIBentryALTinterwordstretchfactor\fontdimen3\font minus
  \fontdimen4\font\relax}
\providecommand{\BIBforeignlanguage}[2]{{%
\expandafter\ifx\csname l@#1\endcsname\relax
\typeout{** WARNING: IEEEtran.bst: No hyphenation pattern has been}%
\typeout{** loaded for the language `#1'. Using the pattern for}%
\typeout{** the default language instead.}%
\else
\language=\csname l@#1\endcsname
\fi
#2}}
\providecommand{\BIBdecl}{\relax}
\BIBdecl

\bibitem{KuVa86}
P.~R. Kumar and P.~Varaiya, \emph{Stochastic systems: estimation,
  identification, and adaptive control}.\hskip 1em plus 0.5em minus 0.4em\relax
  Englewood Cliffs, NJ: Prentice-Hall, 1986.

\bibitem{Wi68}
H.~Witsenhausen, ``A counterexample in stochastic optimum control,'' \emph{SIAM
  Journal on Control}, vol.~6, no.~1, pp. 131--147, 1968.

\bibitem{HoCh72}
Y.~C. Ho and K.-H. Chu, ``Team decision theory and information structures in
  optimal control problems--part i,'' \emph{Automatic Control, IEEE
  Transactions on}, vol.~17, no.~1, pp. 15--22, 1972.

\bibitem{Yu09}
S.~Y{\"u}ksel, ``Stochastic nestedness and the belief sharing information
  pattern,'' \emph{Automatic Control, IEEE Transactions on}, vol.~54, no.~12,
  pp. 2773--2786, 2009.

\bibitem{MaNa15}
A.~Mahajan and A.~Nayyar, ``Sufficient statistics for linear control strategies
  in decentralized systems with partial history sharing,'' \emph{IEEE
  Transactions on Automatic Control}, vol.~60, no.~8, pp. 2046--2056, Aug 2015.

\bibitem{OsRu94}
M.~J. Osborne and A.~Rubinstein, \emph{A Course in Game Theory}, ser. MIT Press
  Books.\hskip 1em plus 0.5em minus 0.4em\relax The MIT Press, 1994, vol.~1.

\bibitem{FuTi91}
D.~Fudenberg and J.~Tirole, \emph{Game Theory}.\hskip 1em plus 0.5em minus
  0.4em\relax Cambridge, MA: MIT Press, 1991.

\bibitem{Ba78}
\BIBentryALTinterwordspacing
T.~Ba{\c s}ar, ``Two-criteria {LQG} decision problems with one-step delay
  observation sharing pattern,'' \emph{Information and Control}, vol.~38,
  no.~1, pp. 21 -- 50, 1978.

\bibitem{NaGuLaBa14}
A.~Nayyar, A.~Gupta, C.~Langbort, and T.~Ba{\c s}ar, ``Common information based
  {M}arkov perfect equilibria for stochastic games with asymmetric information:
  Finite games,'' \emph{IEEE Trans.~Automatic Control}, vol.~59, no.~3, pp.
  555--570, March 2014.

\bibitem{Wi71}
H.~S. Witsenhausen, ``Separation of estimation and control for discrete time
  systems,'' \emph{Proceedings of the IEEE}, vol.~59, no.~11, pp. 1557--1566,
  1971.

\bibitem{MaTi01}
E.~Maskin and J.~Tirole, ``{M}arkov perfect equilibrium: I. observable
  actions,'' \emph{Journal of Economic Theory}, vol. 100, no.~2, pp. 191--219,
  2001.

\bibitem{GuNaLaBa14}
A.~Gupta, A.~Nayyar, C.~Langbort, and T.~Ba{\c s}ar, ``Common information based
  {M}arkov perfect equilibria for linear-gaussian games with asymmetric
  information,'' \emph{SIAM Journal on Control and Optimization}, vol.~52,
  no.~5, pp. 3228--3260, 2014.

\bibitem{VaAn16}
D.~Vasal and A.~Anastasopoulos, ``A systematic process for evaluating
  structured perfect {B}ayesian equilibria in dynamic games with asymmetric
  information,'' in \emph{American {C}ontrol {C}onference}, Boston, US, 2016,
  ({A}ccepted for publication), Available on arXiv.

\bibitem{Ho80}
Y.-C. Ho, ``Team decision theory and information structures,''
  \emph{Proceedings of the IEEE}, vol.~68, no.~6, pp. 644--654, 1980.

\bibitem{KrSo94}
D.~M. Kreps and J.~Sobel, ``Chapter 25 signalling,'' ser. Handbook of Game
  Theory with Economic Applications.\hskip 1em plus 0.5em minus 0.4em\relax
  Elsevier, 1994, vol.~2, pp. 849 -- 867.

\bibitem{NaMaTe13}
A.~Nayyar, A.~Mahajan, and D.~Teneketzis, ``Decentralized stochastic control
  with partial history sharing: A common information approach,''
  \emph{Automatic Control, IEEE Transactions on}, vol.~58, no.~7, pp.
  1644--1658, 2013.

\end{thebibliography}

        \end{document}